%% file: NiesScholz_2019.tex
\documentclass[11pt,a4paper, envcountsame]{amsart}
\usepackage{amsaddr}
\usepackage[usenames,dvipsnames]{color}

 \usepackage[colorlinks,citecolor=blue,urlcolor=blue, linkcolor=blue]{hyperref}

 \usepackage{amsmath, amssymb, xspace}
 \usepackage{graphicx}

 \usepackage[all]{xy}

\newtheorem{theorem}{Theorem}[section]

\newtheorem{fact}[theorem]{Fact}

\newtheorem{proposition}[theorem]{Proposition}

\newtheorem{remark}[theorem]{Remark}

\newtheorem{prop}[theorem]{Proposition}

\newtheorem{claim}[theorem]{Claim}

\newtheorem{convention}[theorem]{Convention}

\newtheorem{cor}[theorem]{Corollary}

\theoremstyle{definition}
\newtheorem{definition}[theorem]{Definition}
\newtheorem{conjecture}[theorem]{Conjecture}

\newcommand{\NN}{{\mathbb{N}}}

\newcommand{\sub}{\subseteq}
\newcommand{\sN}[1]{_{#1\in \NN}}
\newcommand{\uhr}[1]{\! \upharpoonright_{#1}}
\newcommand{\ML}{Martin-L{\"o}f}
\newcommand{\SI}[1]{\Sigma^0_{#1}}

\newcommand{\bi}{\begin{itemize}}
\newcommand{\ei}{\end{itemize}}
\newcommand{\bc}{\begin{center}}
\newcommand{\ec}{\end{center}}

\newcommand{\tp}[1]{2^{#1}}
\newcommand{\ex}{\exists}
\newcommand{\fa}{\forall}

\newcommand{\la}{\langle}
\newcommand{\ra}{\rangle}

\newcommand{\seqcantor}{2^{ \NN}}

\newcommand{\cantor}{\seqcantor}

\newcommand{\Opcl}[1]{[#1]^\prec}

\newcommand{\n}{\noindent}

\newcommand{\leb}{\mathbf{\lambda}}

\newcommand{\sss}{\sigma}

\newcommand{\lland}{\, \land \, }

\newcommand \seq[1]{{\left\langle{#1}\right\rangle}}

\newcommand\+[1]{\mathcal{#1}}

\newcommand{\wt}{\widetilde}

\newcommand{\ul}{\underline}

\newcommand{\LR}{\Leftrightarrow}

\newcommand{\sssl}{\ensuremath{|\sigma|}}

\newcommand{\range}{\ensuremath{\mathtt{rg}}}

 \newcommand{\Proj}[2]{{\mbox{\rm \textsf{Proj}}} ( {#1}; {#2})}

 \newcommand{\tr}{\mbox{\rm \textsf{Tr}}}

\newcommand{\Malg}[1]{M_{\tp{#1}}}
\newcommand{\ketbra}[2]{\ensuremath{| #1  \rangle \langle #2 |}}
\newcommand{\ketbradouble}[1]{\ensuremath{| #1  \rangle \langle #1 |}}

\newcommand{\andre}[1]{}

\input{volkherincludes.tex}

\newcommand{\Calg}{\C_{\mathtt{alg}}}
\begin{document}

  \title{Martin-L\"of random quantum states}
  \author{Andr\'e Nies}
    \address{Department of Computer Science, The University of Auckland}
  \email{andre@cs.auckland.ac.nz}
  \author{Volkher B.~Scholz}
  \address{ETH and Belimo, Z\"urich, Switzerland}
  \email{scholz@phys.ethz.ch}
\date{\today}

\maketitle

\begin{abstract}
    We extend the  key notion  of Martin-L\"of randomness for infinite bit sequences   to the quantum setting, where the   sequences become states of an infinite dimensional system.   We prove that our definition naturally extends the classical case. In analogy with the Levin-Schnorr theorem, we  work towards characterising quantum ML-randomness of  states by incompressibility (in the sense of quantum Turing machines)  of all initial segments.   \end{abstract}

\section{Introduction}
\n \emph{Algorithmic theory of randomness in the classical setting.}  An   infinite sequence of classical bits can be thought of as  random if it satisfies no exceptional properties. Examples of  exceptional property are that every other bit is 0, and  that all initial segments  have more   0s than 1s.  An infinite sequences of fair coin tosses has  neither of the two properties.  

 Infinite bit  sequences form the so-called  Cantor space $\cantor$, which is equipped with a natural compact topology,  and the uniform  measure which makes the infinitely many coin tosses independent and fair.  Recall    that a subset of $\cantor$  is defined to be \emph{null} if it is contained in $\bigcap_m G_m$ for a sequence of open sets $G_m$ with measure tending to $0$. An exceptional  property then corresponds  to a null set in   $\cantor$.

Since no sequence can actually avoid all the null sets, one  has to  restrict the   class of null sets that can be  considered. One only allows  null sets that are effective, i.e.\ can be described in an algorithmic way. The possible  levels of effectiveness one can choose    determine  a hierarchy of formal randomness notions.
Such notions are studied  for instance   in the books \cite{Downey.Hirschfeldt:book, Nies:book}. In recent work, the algorithmic  theory of randomness has been  connected to     mathematical fields  such as ergodic theory and set theory~\cite{Nandakumar:08,Hoyrup:12,Monin.Nies:15,Monin.Nies:17}.

 Martin-L\"of (ML) randomness, introduced in~\cite{Martin-Lof:66}, is a central algorithmic   randomness notion.  Roughly speaking, a bit sequence $Z$ is ML-random if it is in no null set $\bigcap_{m\in \NN}  G_m$ where the $G_m$ are effectively open uniformly in $m$, and the uniform measure of $G_m$ is at most $\tp {-m}$. This  notion is central because there is a universal test, and because ML-randomness  of an infinite  bit sequence can be naturally characterised  by an incompressibility condition on the initial segments of the sequence  (Levin-Schnorr theorem). Detail will be given below.

\n \emph{The quantum setting.} Our main goal is to develop an algorithmic  theory of randomness for infinite sequences of quantum bits. This poses two challenges.

The first challenge  is to provide a  satisfying mathematical model for such sequences. This is not as straightforward as in the classical case: deleting one qubit from    a system of finitely many entangled qubits (e.g.\ the EPR state, which describes two entangled photons) creates a mixed state, namely a statistical superposition  of possibilities for the remaining qubits.  So one  actually  studies  states of  a system that can be interpreted as statistical superpositions of   infinite sequences of   quantum bits (qubits).    Such states have  been considered    in theoretical physics in the form of  half-infinite spin chains (e.g.\ a linear arrangement  of  hydrogen atoms with the electron in the basic or the excited state).

 The  usual mathematical approach (e.g.\ \cite{Bratteli.Robinson:12,Bjelakovic.etal:04,Benatti.etal:14}) to deal with such statistical superpositions is as follows.  The sequences of qubits are  modeled by coherent sequences $\seq{\rho_n }\sN n$ of density matrices. We have  $\rho_n \in M_n$ where   $M_n$ is  the algebra of $ \tp n \times \tp n$ matrices over $\mathbb C$. The idea is that  $\rho_n$    describes the first $n$ qubits. Infinite qubit sequences  are  states (that is, positive functionals of norm~$1$) of a certain  C$^*$-algebra  $\Malg \infty$, the direct  limit of the  matrix algebras $M_n$. 

The second  challenge is the absence of measure  in the quantum setting. We will use instead  the unique   tracial state $\tau$ on $M_\infty$ as  a noncommutative analog of  the uniform  measure. For a projection $p \in M_n$   one has $\tau(p)= \tp{-n} \dim \range (p)$. The analog of an effectively open set in Cantor space is now  a computable  increasing sequence $G = \seq {p_n} \sN  n$ of projections, $p_n \in M_n$, and one defines $\tau(G) = \sup_n \tau(p_n)$. Based on this we will introduce our main technical concept, a quantum version of ML-tests.

 \n \emph{Overview of the paper.}   Section~\ref{s:prelim}   provides  the necessary preliminaries on \emph{finite} sequences of qubits,  as well as  on density matrices, which describe statistical superpositions of qubit  sequences of the same length.   We also    review the mathematical model for states that embody  infinite sequences of qubits.

 In Section~\ref{s:qMLrd}  we introduce    quantum \ML\ tests. We show that there is  a universal such  test. Every infinite   sequence of classical bits can be seen as a state of the C$^*$-algebra~$\Malg \infty$. We show that for such a  sequence, quantum ML-randomness coincides with the usual ML-randomness.  So our notion naturally extends the classical one.

The Levin-Schnorr  theorem (Levin~\cite{Levin:74}, Schnorr~\cite{Schnorr:71})   characterises ML-randomness of a bit sequence $Z$ by the growth rate of the initial segment  complexity $K(Z \uhr n)$ (here $Z\uhr n $ denotes  the string consisting of the first $n$ bits, and $K$ denotes a version of Kolmogorov complexity where the set of   descriptions that   a universal machine can use as   inputs has to be prefix free).  In Section~\ref{s:char} we work towards  a potential quantum  version of  this important result. This would mean   that quantum \ML\ random states are characterized by  having  initial segments of  a fast growing  quantum Kolmogorov complexity. The actual formulation of our result corresponds to the G\'acs-Miller-Yu theorem~\cite{Gacs:80,Miller.Yu:08} which uses plain Kolmogorov complexity~$C$,  rather than the original  Levin-Schnorr theorem,  for reasons related to different properties of classical and quantum Turing machines~\cite{Berthiaume.ea:00}.

We note that there has been an earlier application of notions from computability theory  to spin chains.  Wolf, Cubitt and Perez-Garcia~\cite{Cubitt.etal:15} studied  undecidability in the quantum setting. They constructed Hamiltonians on  square lattices with associated ground states which radically change behaviour as the system size grows. For example, while being a product state for small system sizes, it becomes entangled for large sizes~\cite{Bausch.etal:18}. Furthermore,  based only on initial segments of the sequence, it is computationally impossible to predict whether this effects happens. The states we consider here are defined on a spin chain rather than on a two-dimensional lattice, and are also not constrained to be ground states of local Hamiltonians. However, they share similar features in the sense of possessing unpredictable behaviour as the system size grows.



 \section{Preliminaries} \label{s:prelim}
 
\subsection{Quantum bits.}   A  classical bit can be in states $0,1$.      A qubit   is a  physical system with two possible classical states: for instance, the          polarisation of photon horizontal/vertical,   an  hydrogen atom in the ground or the first   excited state.      
     A qubit can be  in a superposition of the two classical states: 
  $\alpha \mid 0\ra + \beta \mid 1 \ra,$ 
  where   $\alpha, \beta \in \mathbb C $,  $|\alpha|^2 + |\beta|^2 =1$. 
        A measurement of a qubit  w.r.t.\ the standard basis $|0\ra, |1 \ra$ yields $0$ with probability $|\alpha|^2$,
   and $1$ with probability  $|\beta|^2$. 

    \smallskip
    
\subsection{Finite sequences of quantum bits.}
  The state of a physical system is represented by a vector in  a  finite dimensional Hilbert space~$A$.   For vectors $a,b \in A$, 
  $\la a | b \ra$ denotes the inner product of vectors $a,b$, which is  linear in the \emph{second} component and antilinear in the first.    
 For systems represented by Hilbert spaces  $A,B$, the tensor product $A \otimes B$ is a Hilbert space that represents the combined system. 
 One  
 defines an inner product on $A \otimes B$ by 
  $\la a\otimes b  | c \otimes d \ra = \la a | c \ra \la b | d \ra$.

Mathematically,  a qubit is simply  a unit vector in $\mathbb C^2$.  
The state of a  system of   $n$ qubits is   a  unit vector in the tensor power 
\bc $\+ H_n : = (\mathbb C^2)^{\otimes n} =  \underbrace   {\mathbb C^2 \otimes  \ldots \otimes \mathbb C^2}_{n}$. \ec    
    We   denote the standard basis of $\mathbb C^2$ by $| 0 \ra, | 1 \ra$.
The standard basis of $\+ H_n $   consists of vectors 
\bc $  | a_0 \ldots a_{n-1} \ra : = | a_0 \ra \otimes \ldots \otimes | a_{n-1} \ra$, \ec    
 where $\sigma = a_0\ldots a_{n-1}$ is an $n$-bit string. 
 While the usual notation for $| a_0 \ldots a_{n-1} \ra$  in quantum physics is   $| \sss \ra$, for brevity   we  will often write $\ul \sss$. 
     The state of the system of $n$ qubits is a unit vector  $\+ H_n $   and hence a  certain linear superposition of these basis vectors.
For example,  for $n=2$, the   EPR (or ``maximally entangled")  state is $\frac 1 {\sqrt 2} ( |00\ra  + |11 \ra)$.
\smallskip
    
\subsection{Mixed states and density operators.} 
With each 
 state $|\psi \ra$ (a  unit vector in $\+ H_n$) we associate, and often identify,  the linear form     $|\psi \ra  \la \psi |$ on $\+ H_n$ 
 given by $\ket{\phi}    \mapsto  |\psi \ra \la \psi \mid \phi \ra$. A \emph{mixed state} is a convex linear combination $\sum_{i=1}^{2^n} p_i  |\psi_i \ra  \la \psi_i |$ of    pairwise orthogonal   states $\psi_i$. 
 In this context a  state  $|\psi \ra$  (identified with $\ketbradouble \psi$) is called \emph{pure}.

    Recall that for an operator $S$ on a finite dimensional Hilbert space $A$, the  \emph{trace}  $\tr (S)$ is the  sum of the eigenvalues of $ S$ (counted with multiplicity).  A Hermitian operator is called positive if  all eigenvalues are non-negative.
   A mixed state corresponds to  a positive  operator  $S$  on $\+ H_n$  with $\tr (S) =1$, as one can see  via  the spectral decomposition of $S$.

  A $C^*$-algebra is a subalgebra of the bounded operators on some Hilbert space closed under taking the adjoint, and topologically closed in the operator norm.
We let 
\bc $M_k= \mathrm{Mat}_{2^k}(\C)$ \ec  denote the $C^*$-algebra of $\tp k\times \tp k$
matrices over $\C$ (identified with operators on $\C^{\tp k}$). 
A~\emph{density operator} (or density matrix)  in $ \Malg k$  is a positive  operator in $M_k$ with trace 1. The states $\rho \in \S(M_{k})$ can be  identified with the density operators $S$ on $\+ H_k$: to  a state $\rho$ corresponds     the unique density operator $S$ such that 
  \begin{equation}\label{eqn:DO}
  \rho(X) = \tr (SX) \text{  for each }X\in M_{k}.
  \end{equation}

\smallskip
    
\subsection{Embeddings between matrix algebras.}   We view $\+ H_{n+1} $ as the tensor product $ \+ H_n \otimes \C^2$. Then $\Malg {n+1} $ is naturally isomorphic to $ \Malg n \otimes \Malg 1$.   We view the indices of matrix entries  as numbers written in ``reverse binary'', i.e., with the most significant digit written on the right.  Thus a matrix entry is indexed by a pair  of strings  $\sss, \tau$ of the same length, and  a matrix in $\Malg {n+1}$ has the form
\[ A =  \left ( \begin{matrix} A_{00}  &  A_{01}  \\ A_{10}  &  A_{11}  \end{matrix} \right ) \]
where  each $A_{i,k} = (a_{\sss i , \tau k})_{|\sss| = |\tau| =n}$ is in $\Malg n$.
 We have embeddings $\Malg n \to \Malg {n+1}$ given by \begin{equation} \label{eqn:embed} A \mapsto A \otimes I_2 = \left ( \begin{matrix} A  &  0 \\ 0 &  A \end{matrix} \right ). \end{equation} 
Note that the embeddings preserve the operator norm.

\smallskip

\subsection{Partial trace operation.} \label{ss:pT} For each $n$, there  is a  unique linear map $T_n \colon \Malg {n+1}   \to \Malg n$, called  the  partial trace operation, such that $T_n( R \otimes S) = R \cdot  \tr(S)$ for each $R \in \Malg n$ and $S \in \Malg 1$. 
Intuitively, this operation   corresponds to deleting the last qubit; for instance,  $T_1 ( \ketbra  {10} {10}) = \ketbra 1 1 $.    
   \begin{remark} \label{rem:EPR} {\rm Consider again the EPR  state $\frac 1 {\sqrt 2}( |00\ra  + |11 \ra)$, now viewed as the   operator  $\beta = \frac 1 { 2}( |00\ra  + |11 \ra )  ( \la 00 | + \la 11 |)$ in $M_2 $.  While this state is pure,   $T_1(\beta) = \frac 1 2 (\ketbra {0} 0  +   \ketbra 1 1 )$   is a properly mixed state.  }  \end{remark}

 One can provide an explicit description of the partial trace operation $T_n$  as follows. $\+ H_{n+1}$ has as a base the $ |{\sss r}\ra $, $\sss$ a string of $n$ bits, $r$ a bit. For    a $\tp{n+1} \times \tp{n+1}$ matrix $A= ( a_{\sss  r, \tau s})_{\sssl, |\tau|=n, r,s \in \{0,1\}}$,  $T_n(A)$ is   given by the $\tp n \times \tp n$  matrix  \begin{equation} \label{small matrix} b_{\sss, \tau} = a_{\sss 0, \tau 0} + a_{\sss 1 , \tau 1}. \end{equation}  
It is easy to check that if $A$ is a density matrix, then so is $T_n(A)$. 

%
%

\subsection{Direct limit of matrix algebras, and tracial states.} \label{ss:TS} The so-called CAR algebra $\Malg {\infty}$  is the direct limit of the $\Malg k $ under the norm-preserving embeddings in~(\ref{eqn:embed}).  Thus, $\Malg \infty$ is the norm completion of the   union of the $\Malg n$, seen as a $*$-algebra. Clearly $\Malg \infty$ is a $C^*$-algebra. Note that it is more  common to write $\Malg {\tp\infty}$ for this algebra, and to denote our $M_k$'s by $M_{\tp k}$, but we use the present notations for simplicity.

 A \emph{state} on a   $C^*$-algebra $M$ is a positive linear functional $\rho \colon M \to \C$ that sends the unit element of $M$ to  $1$ (this implies that $|| \rho || = 1$). To be positive means that $x \ge 0 \to \rho(x) \ge 0$.

A state $\rho$ is called \emph{tracial} if $\rho(ab) = \rho (ba)$ for each pair of operators $a,b$. On $\Malg n$ there is a unique tracial  state $\tau_n$ given by $\tau_n (a) = \tp{-n} \tr(a)= \tp{-n} \sum_{\sssl =n} a_{\sss, \sss}$. The corresponding density matrix is $\tp{-n} I_{\tp n}$ (i.e., it has $\tp{-n}$ on the diagonal and $0$ elsewhere). Note that the states $\tau_n$  
are compatible with the embeddings $M_n \to M_{n+1}$. This yields a tracial state    $\tau$ on $\Malg {\infty}$,  which is unique as well.   


\subsection{Quantum   Cantor space.} \label{s:qCantor}
  \label{ex:measure} \label{ex:diracmeasure}
    The quantum analog of Cantor space is $\+ S(\Malg \infty)$,  the set of states  of the $C^*$-algebra $\Malg \infty$.  The space  $\+ S(\Malg \infty)$ is endowed with  a convex structure, and   is compact in the weak $*$ topology (the coarsest topology that makes the application maps $\rho\mapsto \rho(x)$ continuous) by the Banach-Alaoglu theorem. The following is well known but hard to reference in this form. 
\begin{fact}  A state  $\rho \in  \+ S(\Malg \infty)$ corresponds to a   sequence $\seq {\rho_n} \sN n$, where  $\rho_n$ is a  density matrix in $\Malg n$, and   which is coherent in the sense that taking the  partial trace of $\rho_{n+1}$ yields $\rho_n$.   \end{fact}

 \begin{proof}  First let  $\rho \in  \S(\Malg {\infty})$.  Let $\rho_n $ be the state which is the  restriction  of $\rho$ to $\Malg n $ (later on  we will use the notation $\rho \uhr n$).  Let   $S_n$ be the  density matrix on $\Malg n$ corresponding to $\rho_n$  according to~(\ref{eqn:DO}).        
\begin{claim}  \label{fact:coherence}   $T_n(\rho_{n+1}) = \rho_n$. \end{claim}
To see this note that a brief calculation using   (\ref{small matrix}) shows that  for each $A \in \Malg n$   \[ \tr  \left( S_{n+1}  \left ( \begin{matrix} A  &  0 \\ 0 &  A \end{matrix} \right ) \right ) = \tr \  ( T_n(S_{n+1}) A). \]
We also have  $\rho_{n+1} (A \otimes I_2)  = \rho_n(A)$ by definition of the $\rho_n$. Therefore  $T_n(S_{n+1}) = S_n$ by the uniqueness of the density matrix for $\rho_n$.

Conversely, given a sequence $(\rho_n)$ of states on $\Malg n$  such that $T_n(\rho_{n+1}) = \rho_n$, there is a unique $\rho  \in  \S(\Malg {\infty})$ such that  $\rho_n $ is  the restriction of $\rho$ to $\Malg n $ for each~$n$.  To see this, first one defines a bounded functional $  \wt \rho$ on the $*$-algebra $\bigcup_n \Malg n$ that extends each $\rho_n$. Then one extends $\widetilde \rho$   to  a state $\rho$ on  $\Malg {\infty}$ using that $\widetilde  \rho$ is continuous.    \end{proof}
  
 We  need to allow mixed states as sequence entries when we develop an    analog   of infinite bit sequences for qubits: as discussed in Remark~\ref{rem:EPR},  the EPR state $\beta \in M_2$ is a pure state which turns into the  mixed state after taking the  partial trace.   The following example may be instructive: letting  $\rho_{2n}=\beta^{\otimes n}$ and $\rho_{2n+1}=\rho_{2n}\otimes \frac 1 2 (\ketbra {0} 0  +   \ketbra 1 1 ) $,  we obtain a state such that the even initial segments are pure and the odd ones are mixed.

 \begin{remark} \label{rem: measure} {\rm Suppose that all the $\rho_n$ are in diagonal form, and hence the entries of the corresponding matrices are in $[0,1]$. For each $\sss$ we can interpret $a_{\sss, \sss}$ as the probability that $\sss$ is an initial segment of a bit sequence: by (\ref{small matrix}) we have $a_{\sss, \sss} = a_{\sss 0, \sss 0} + a_{\sss 1 , \sss 1}$. In other words, the  $\rho  \in \+ \S(M_{2^\infty})$ with all $\rho_n$  in diagonal form correspond to the probability measures on $\cantor$. The tracial state $\tau$ defined in Subsection~\ref{ss:TS}  corresponds to the uniform measure. } \end{remark}

One can now view an infinite sequence of classical bits $Z \in \cantor$ as a state in $\S(M_{2^\infty})$,  which corresponds to the probability measure concentrating on $\{ Z \}$.    For  more detail,     recall that the Hilbert space $\+ H_{n}$ has as a base the vectors   $| \sss  \ra $, for a string $\sss$  of $n$ classical bits. 
   A  classical bit sequence $Z$ corresponds to the state  $(\rho_n)\sN n$,  where the bit  matrix $B=\rho_n \in \Malg {n}$  is given by  $b_{\sss, \tau }= 1  \LR  \sss = \tau = Z \uhr n$. For  $\sss = Z\uhr n$, $\rho_n$  is the pure state $ |   \sss \ra \la    \sss |$  on $M_{n}$.
 
%

\subsubsection*{Characterisation of quantum Cantor space} Cantor space can be characterised  as the projective (or inverse) limit of the discrete topological spaces  $X_n$ of $n$-bit strings with the maps $f_n \colon X_{n+1} \to X_n$ so that $f_n( \sss r) = \sss$, i.e., $f_n$ removes the last bit.   We now show that quantum Cantor space is a projective limit of the state sets of the $\Malg n$.
%
%
%
%

By a  \emph{convex  (topological) space} we mean  a topological space $X$ with a continuous operation  $F(\delta , a,b)=  \delta a + (1-\delta) b$, for $\delta  \in [0,1]$, $a,b \in X$, satisfying     obvious arithmetical axioms such as $F(\delta, a,a)= a$.
  Clearly  $\S (\Malg n)$ is a convex space.  A map $g \colon X \to Y$ between convex   spaces is called \emph{affine}  if $g(F(\delta , a,b)) = F(\delta , g(a),g(b))$ for each $a, b \in X$ and each $\delta$.  The \emph{projective limit} of a sequence $\seq{X_n}$ of convex spaces with continuous affine maps  $T_n \colon X_{n+1} \to X_n$, $n \in \NN$ (called a diagram), is the convex space  $P$ of all $\rho \in  \prod_n X_n$ such that $T_n(\rho(n+1) )= \rho(n)$ for each~$n$ with the subspace topology and the canonical operation~$F$.

Denote by $g_n \colon P \to X_n$ the map sending $\rho $ to $\rho(n)$.  The  projective  limit $P$  is characterised up to affine homeomorphism by as a  colimit from category theory:  we have $T_n  \circ g_{n+1} = g_n$ for each $n$, and  for any convex space $A$ with continous affine maps  $ f_n \colon A \to X_n$ such that $T_n  \circ f_{n+1} = f_n$ for each $n$, there is a unique continuous affine map $f\colon A  \to P$ such that $ g_n \circ f = f_n$ for each $n$.

 \begin{proposition} Consider the diagram consisting of  the $\S (\Malg n)$ together with the partial trace maps $T_n \colon \S(\Malg {n+1}) \to \S(\Malg n)$.    Seen  as a convex  space,  
  $\S(\Malg {\infty})$   is affinely homeomorphic to the projective limit of the $\S (\Malg n)$.  \end{proposition}
\begin{proof} Define $\widehat g_n \colon \S(\Malg {\infty}) \to \S(\Malg {n})$ by $\widehat g_n (\rho) = \rho \uhr {\Malg n}$.  By Claim~\ref{fact:coherence},  $T_n  \circ \widehat g_{n+1} = \widehat g_n$ for each $n$. It now suffices to verify the universal property for $\S(\Malg {\infty}) $  together with the maps $\widehat g_n \colon \S(\Malg {\infty}) \to \S(\Malg {n})$.  Suppose we are given a convex space $A$ with continous affine maps  $ f_n \colon A \to \S(\Malg n)$ such that $T_n  \circ f_{n+1} = f_n$ for each $n$. Let $f\colon A \to \S(\Malg {\infty}) $ be the  map such that $f(x) = \rho$ where $\rho $ is the state determined  as above by the sequence $\rho_n \in \S(\Malg {n})$ such that $\rho_n  = f_n(x)$. Clearly $f$ is affine. We  show that $f $ is continous. A basic open set  of $\S (\Malg \infty)$ with its weak-$*$ topology has the form \bc $U_{v, S} = \{ \rho \colon \, \rho(v) \in S\} $ \ec 
where $v \in \Malg \infty$ and $S \sub \mathbb C$ is open.  If $\rho \in U_{v, S} $ then there is $\varepsilon>0$ such that the open ball $B_\varepsilon (\rho(v))$ is contained in $S$. Furthermore, there  is  $k \in \NN$ and  $w \in  \Malg k$    such that $\norm  {v-w} < \varepsilon/2$. Since states have operator norm $1$, this implies $| \rho(w) - \rho(v)| < \varepsilon/2$. Letting $S'  = B_{\varepsilon/2}(\rho(v))$, we have $U_{w, S' } \sub U_{v, S}$.  

So for continuity of $f$ it suffices to show that $f^{-1}(U_{w, S'})$ is open for any $k$, any $w \in \Malg k$ and open $S' \sub \mathbb C$. By definition of $f$ we have $f^{-1}(U_{w, S'})=f_k^{-1}(\{ \theta\in \+ S(M_k) \colon \theta (w) \in U_{w, S'}\})$. Since $\rho(w) = \rho\uhr{\Malg k} (w)$,      $f^{-1}(U_{w, S'})$ is open  by the hypothesis that $f_k \colon A \to \S(\Malg k)$ is continous.  
\end{proof}

\section{Randomness for states    of $\Malg \infty$}  \label{s:qMLrd}

Our main purpose is  to introduce  and study  a   version of \ML's randomness notion for states on $\Malg \infty$,   Definition~\ref{df:qML}  below.
 We begin by recalling the definition of  \ML\  tests, phrased in such a  way that it can be easily  lifted to the quantum setting. 
 A clopen (i.e.\ closed and open) set $C$ in Cantor space can be  described  by a   set $F$ of  strings  of the same length~$k$ in the sense that $C = \bigcup_{\sss \in F} \{ Z \colon Z \succeq \sss\}$ (note that $k$  is not unique). A $\SI 1$ set $\+ S$ (or effectively open set) is a subset of   Cantor space $\cantor$  given in a computable way as   an ascending  union of clopen sets. In more detail, we have $\+ S = \bigcup C_n$ where $C_n$ is clopen, $C_n \sub C_{n+1}$,  and a finite description of $C_n$   can be computed from $n$. A  \ML\ \emph{test} is a     uniformly computable  sequence of $\SI 1$ sets $\seq {G_m}\sN m$  (i.e., $ G_m = \bigcup_k C^m_k$, where  the map sending a pair $m,k $  to a description of the clopen set $C^m_k$ is computable)
such that $\leb G_m  \le \tp {-m}$. Here  $\leb$ denotes the uniform measure, which ÊCyrilÊis obtained by viewing the $k$-th bit  as the result of the $k$-th toss of a fair coin, where all the  coin tosses are independent.  

  A sequence $Z\in \cantor$ is \ML\ random if it passes all such tests in the sense that $Z \not \in \bigcap_m G_m$. By the 1973 Levin-Schnorr theorem (see e.g.\ \cite[3.2.9]{Nies:book}), this is equivalent to the incompressibility condition on initial segments that for some   constant $b$,  for each $n$,   the prefix free Kolmogorov complexity of  the first $n$ bits of $Z$  is at least $n-b$. 

As an aside, if we add the restriction on tests that the measure of $G_m$ is  a computable real uniformly in $m$, we obtain the weaker notion of Schnorr randomness, now frequently used in the effective study of theorems from  analysis, e.g. \cite{Pathak.Rojas.ea:12}. This notion (as well as its variant, computable randomness) embody an alternative paradigm of randomness, namely that it is hard to predict the next bit from the previously seen ones. 


  In the following we will use    letters $\rho, \eta$  for    states on  $\Malg \infty$. We  write $\rho \uhr n$ for the restriction of $\rho$ to $\Malg n$, viewed either as a density matrix or a state of $\Malg n$.
After introducing    quantum ML-randomness   in  Definition~\ref{df:qML}, we will show that it ties in with the classical definition of ML-randomness. As mentioned above, any classical bit sequence defines a pure (product) quantum state of $\Malg \infty$, by mapping the bits to the corresponding basis elements in the computational basis. We then prove that for classical bit sequences, \ML\  randomness agrees with its quantum analog under this embedding. Even in the classical setting our notion  is broader, because  probability measures over  infinite  bit sequences can be viewed as states according to Remark~\ref{rem: measure}. For instance, the uniform measure on $\cantor$, seen as the tracial state $\tau$,  is quantum ML-random. So in the new setting, randomness of $\rho$ does not contradict that the function $n \to \rho \uhr n$ is computable.  Philosophically   there is some doubt  whether these states should be called random at all; the  term  ``unstructured"   appears   more apt, and only for classical bit sequences being unstructured actually implies being random. However, we   prefer the term ``random" here simply for  practical reasons.

%

%

\subsection{Quantum analog of Martin-L\"of tests}

Quite generally, a \emph{projection} in a $C^*$ algebra is a self-adjoint  positive operator $p$ such that  $p^2 = p$.  In the definition of ML-tests,  we will replace a clopen set  given by strings of length $n$ by a projection in $\Malg n$. However, we need to restrict its possible matrix entries to complex numbers that have  a finite description.

  \begin{definition} \label{def:Calg}  A complex number  $z$ is called \emph{algebraic} if it is the root of a polynomial with rational coefficients.
 Let $\Calg$ denote the  field of algebraic complex numbers. A matrix over $\Calg$ will be called \emph{elementary}.  \end{definition} Clearly such numbers  are given by a finite amount of information, and in fact they are polynomial time computable: for each $n$ one can in time polynomial in $n$ compute a Gaussian rational within $\tp{-n}$ of $z$.  Note that if the matrix determining an  operator on $\+ H_n$  consists of  algebraic complex numbers, then  its eigenvalues are in $\Calg$ and its  eigenvectors are vectors over $\Calg$. 
 We note that by a result of Rabin, $\Calg$ has a computable presentation:  there is  a 1-1 function
$f\colon \Calg \to \NN$ such that the image under  $f$ of the field operations are partial computable functions with computable domains.

 Suppose that a projection $p \in \Malg n$ is   diagonal with respect to the standard basis.  Since its only possible eigenvalues are $0,1$, the entries are $0$ or $1$. Thus,    projections in $M_n$ with a diagonal matrix  directly correspond to clopen sets in Cantor space.

 For $u,v \in \Malg n$ one writes $u\le v$ if $v-u$ is positive.  Note that $p \le q$ for projections   $p,q \in \Malg n$ means that the range of $p$ is contained in the range of~$q$.  A projection $p \in M_{n}$   with  matrix   entries in $\Calg$  will be called a  \emph{special projection}.   Such a projection   has a finite description,  given by the size of its matrix and all the entries in  its  matrix.     
 \begin{definition} \label{df:quantum S1} A quantum $\SI 1$ set   (or q-$\SI 1$ set for short) $G$  is   a computable    sequence of special  projections $ (p_i)\sN i$  such that  $p_i \in \Malg i$ and  $p_i \le p_{i+1}$ for each $i$.   \end{definition}
 We note that the limit of an increasing sequence of projections  $\seq{p_i} \sN i$ does not necessarily exist in $\Malg \infty$. For,  the limit would be a projection itself, and the projection lattice of $\Malg \infty$ is not complete. See e.g.\ \cite{Furber:19} where it is shown that its completeness would yield an embedding of the non-separable $C^*$-algebra $\ell^\infty(\mathbb C)$ into $\Malg \infty$, which  is not possible as the latter is separable. 
 
 Recall that for a state $\rho$ and $p \in \Malg n$ we have $\rho(p) = \tr (\rho\uhr n p)$ by (\ref{eqn:DO}). We write $\rho(G) = \sup_i \rho(p_i)$.   In particular we let $\tau(G) = \sup_i  \tau(p_i)$, where $\tau $ is the tracial state defined in Section~\ref{ss:TS}. Note that $\tau(G)$ is a   real in $[0,1]$ that is the supremum of a computable sequence of rationals. Each $\SI 1$ set  $\+ S$ in Cantor space is an effective union of clopen sets,  and hence  can be seen as a q-$\SI 1$ set. If the state $\rho$ is a measure then  $\rho(\+ S)  $    yields the  usual result: evaluating $\rho$ on $\+ S$.

We chose the term ``quantum $\SI 1$ set'' by analogy with the notion of $\SI 1$ subsets of Cantor space; they are  not actually  sets. The physical intuition is that a projection $p_i \in \Malg i$ describes a measurement of $\rho \uhr {\Malg i}$, strictly speaking given by the pair of projections $(p_i, I_{\Malg i} -p_i)$. Then  $\rho(p_i) = \tr (\rho \uhr {\Malg i }  p_i)$ is the outcome of the measurement, the  probability that the first alternative given by the measurement occurs, and $\rho(G)$ is the overall outcome  of measuring the state.  So one could view $G$ as a probabilistic set of states: for $0\le \delta \le 1$, $\rho$ is in    $G$  with probability  $\delta$ if $\rho(G) > \delta$. The ``inclusion relation" is $G \le H$ if  $\rho(G) \le  \rho(H)$  for each state $\rho$.  We note that a  ``level set'' of states $\{ \rho \colon \,  \rho(G) > \delta \}$ is open in $\+ S(\Malg \infty)$ with the weak $*$ topology (Section~\ref{s:qCantor}).

%
%
%
For special projections  $p, q \in \Malg k$, we  denote by $p \vee q$ the projection   in $M_k$ with range   $\range \, p + \range \, q$.    We have $\tau(p \vee q) \le \tau(p) + \tau(q)$.
   For quantum $\SI 1$ sets $G = \seq {p_k}\sN k$ and $H = \seq {q_k}\sN k$ we define $G \vee H$ to be $\seq {p_k \vee q_k}\sN k$. Note that again $\tau(G \vee H) \le \tau(G) + \tau(H)$.  Inductively  we then  define $G_1 \vee \ldots \vee G_r$ for $r \ge 2$.

%
   
%
\begin{definition}[Quantum Martin-L\"of randomness] \label{df:qML} A \emph{quantum Martin-L\"of test}  (qML-test)  is an effective  sequence $\seq {G_r}\sN r$ of   quantum $\SI 1$ sets such that $\tau(G_r) \le \tp{-r}$ for each $r$. 
%
For  $\delta \in (0,1)$, we     say that $\rho$ \emph{fails} the qML test \emph{at order} $\delta$   if   $ \rho(G_r) > \delta$ for each $r$; otherwise $\rho$ \emph{passes} the qML test \emph{at order} $\delta$. 
We say that $\rho$ is \emph{quantum ML-random}  if it passes each qML test  $\seq {G_r}\sN r$ at each positive order, that is,  $\inf_r \rho(G_r)=0$. 
 \end{definition}

 \begin{proposition} \label{prop:uML} There is a   qML-test $\seq {R_n}$ such that for each qML test $\seq {G_k}$ and each state $\rho$, for each $n$ there is $k$ such that  $\rho(R_n) \ge  \rho(G_k)$.  In particular, the test is universal in the sense that $\rho$ is qML random iff $\rho$ passes this single test.  \end{proposition} 
 
\begin{proof} This follows the usual construction due to \ML; see e.g.\ \cite[3.2.4]{Nies:book}. We may fix an   effective listing   $\seq{G^e_m}\sN m$ ($e \in \NN$) of all the quantum ML tests, where $G^e_m = \seq{ p_{m,r}^{e}}\sN r$ for projections $p_{m,r}^{e}$ in  $\Malg r$.  Informally, we let $$R_n = \bigvee_{e} G^e_{e+n+1}.$$  However,  this infinite supremum of quantum $\SI 1$ sets  is actually not defined. To interpret it,  think of 
 $ \bigvee_{e} G^e_{e+n+1} $ as $ \bigvee_{e} \bigvee_r p_{e+n+1,r}^{e}$ (which is still not defined).  Now let 
\bc  $q^n_k = \bigvee_{e+n+1 \le k} p_{e+n+1,k}^{e}$ \ec 
which is a finite supremum of projections.  Clearly $q^n_k \le q^n_{k+1}$, and $\tau (q^n_k) \le \sum_e \tau (p_{e+n+1,k}^{e}) \le \tp{-n}$.  We   let  $R_n = \seq {q^n_k} \sN k$, so that $\tau(R_n) \le \tp{-n}$.  Hence  $\seq {R_n}\sN n$ is a quantum ML-test.

Fix  $e$. For each $n$ we have 
\begin{eqnarray*}  \rho(R_n)  =  \sup_k \rho(q^n_k ) &  \ge &   \sup_k  \rho(p^e_{n+e+1,k})  \\ 
& = &  \rho(G^e_{n+e+1}).\end{eqnarray*}
%
%
 \end{proof}

%

%
%
%

A basic property one  expects of random bit sequences $Z$ is the   law of large numbers, namely $\lim_{n \to \infty}  \frac 1 n \sum_{i=0}^{n-1}Z(i) =1/2$. This property holds for \ML\, (or even Schnorr) random bit sequences; see e.g.\ \cite[3.5.21]{Nies:book}.  The first author and Tomamichel have proved a version of the  law of large numbers for qML-random states. For $i<n$ let $S_{n,i}$ be the subspace of $\mathbb C^{2^n}$ generated by those $\underline \sigma$ with $\sigma_i =1$. It is as usual identified with its orthogonal projection. So for any state $\rho$ on $\Malg \infty$, the real $\rho(S_{n,i}) = \tr (\rho\uhr n  S^n_i)$ is the probability that  a  measurement of the $i$-th qubit of its initial segment $\rho\uhr n$ returns $1$. 

\begin{prop}[\cite{LogicBlog:17}, Section 6.6]  Let $\rho$ be a qML-random state. We have \bc $\lim_n \frac 1 n \sum_{i< n}  \rho (S_{n,i}) = 1/2$. \ec    \end{prop}
Their argument is based on Chernoff bounds. It   works in more generality for any  computable bias $r$ in place of  1/2, and  for states that are qML-random with respect to that bias, as detailed in \cite{LogicBlog:17}, Section 6.6.  

\subsection{Comparison with ML-randomness for  bit sequences}
 Recall from  Subsection~\ref{ex:diracmeasure} that each bit sequence $Z$ can be viewed as a state on $\Malg \infty$. In this   section we  show that $Z$ if  ML-random iff $Z$ viewed as a state  is qML-random (Thm.\ \ref{prop:ML bit sequence}).   Each classical ML-test can be viewed as  quantum ML-test, so quantum ML-randomness implies ML-randomness for $Z$. For the converse implication, the idea is to turn a quantum ML-test that $Z$ fails at order $\delta$    into a classical test that $Z$ fails.  We need a few  preliminaries. We thank the anonymous    referees for suggesting simplifications implemented in  the argument below.

Recall that the    vectors $\ul \sss$,  for $n$-bit strings $\sss$, form   the standard basis of~$\+ H_n$. Note that if   $p \in \Malg k$ is  a projection and $\eta$ is a bitstring of length $k$,  then $\tr ( \ketbra{\ul \eta}{  \ul \eta } p) = ||p(\ul \eta) ||^2 = \la \ul \eta | p | \ul  \eta\ra$.  Given a bit sequence $Z$, letting  $\eta = Z\uhr k$, we have    $  Z(p) = \tr (\ketbra {\ul \eta} {\ul \eta} p)$ (if $Z$ is viewed as a state then   $Z\uhr k$ is viewed as the   density matrix    $| \ul \eta \ra \la \ul \eta |$  in Dirac notation).  
  \begin{definition} Fix $k\in \NN$, and let    $p \in \Malg k$ be a projection. For $\delta > 0$ define
	\begin{equation} \label{eqn:S}  S= S^k_{p, \delta} = \{  \eta\in \{0,1\}^k  \colon \,  \delta \le   \tr ( \ketbra{\ul \eta}{  \ul \eta } p) \}.\end{equation} \end{definition}
%
	 

In the following we identify the  set $S$ of strings of length $k$  in (\ref{eqn:S}) with the corresponding diagonal projection in $\Malg k$. By $|S|$ we denote the size of the set $S$. 
	\begin{claim}  \label{claim:boundS} $\tau(S) \le \tau(p) / \delta$.
	\end{claim}
\begin{proof}   
 $\delta |S| \le \sum_{\eta \in S} \tr ( \ketbra {\ul \eta } {\ul \eta} p) \le    \sum_{\eta } \tr ( \ketbra {\ul \eta } {\ul \eta} p)  = \tr (p)$,  
 
 \n   so
 $|S| \tp{-k} \le \tr(p) \tp{-k} / \delta = \tau(q)/ \delta$.   
\end{proof}

 
	\begin{claim}  \label{claim:next k}  Suppose  $p\in \Malg	k$ are as above. Then \[ S^{k+1}_{p', \delta} =  \{ \eta a  \colon \, |\eta|=k, a=0,1  \lland \eta \in S^k_{p, \delta} \},\]
	where $p'$ is the lifting of  $p$  to $\Malg {k+1}$.
	\end{claim}
\begin{proof}  
For $\eta, a$ as above we have $p'(\ul{\eta a }) = p(\ul \eta) \otimes \ul a$ and so $\tr ( \ketbra {\ul \eta } {\ul \eta} p) =  ||p(\ul \eta) ||^2=  ||p'(\ul { \eta a}) ||^2 = \tr ( \ketbra {\ul  {\eta a }} {\ul { \eta a}} p')$.
\end{proof}


   \begin{theorem} \label{prop:ML bit sequence} Suppose $Z \in \cantor$. Then $Z$ is ML-random iff   $Z$ viewed as an element of $\S(M_{2^\infty})$ is qML-random. \end{theorem}
      \begin{proof}  Suppose $Z$ fails the qML-test $\seq{G^r}\sN r $ at order  $\delta>0$,  where $G^r$ is given by the sequence $\seq {p^r_k}\sN {k}$. Thus $\fa r \ex k \, Z(p^r_k) > \delta$, and $\sup_k \tau(p^r_k) \le \tp{-r}$.  
      Uniformly in $r$ we will define a $\SI 1$ set $V_r \sub \cantor$ containing $Z$ and  of measure at most $\tp{-r} /\delta$. This will  show that $Z$ is not ML-random.  
      
      We fix $r$ and suppress it from the notation for now. We may assume that $p_k \in \Malg k$ for each $k$.  Define $S^k_{p_k, \delta}$ as in~(\ref{eqn:S}).  We let $V^r= V = \bigcup_k \+ S_k$, where $\+ S_k =    \Opcl{S^k_{p_k, \delta}}$ (here   $\Opcl X $ denotes the open set given by a set $X$ of strings).

Clearly  $Z \in V$.  It remains to verify that  the  uniform measure of $V$ is at most $\tp{-r} /\delta$. By Claim~\ref{claim:boundS},  it suffices to show that $\+ S_k \sub \+ S_{k+1}$ for each $k$. By Claim~\ref{claim:next k}    viewing $p_k \in \Malg{k+1}$,  we can evaluate (\ref{eqn:S}) for $k+1$ and obtain a set of strings generating  the same clopen set $\+ S_k$. 
Since $p_k \le p_{k+1}$, this set of strings is contained in $S^{k+1}_{p_{k+1}, \delta}$.
    \end{proof}
 
%
%
%
The first author and Stephan have studied the case of states that can be seen as measures on Cantor space in a separate paper \cite{Nies.Stephan:19}. They call a measure $\rho$   \ML\- absolutely continuous if $\lim_m\rho(G_m)  = 0$ for each ML-test $\seq{G_m}$.  
Tejas Bhojraj, a PhD student of Joseph Miller at UW Madison, has shown that this notion coincides with quantum ML-randomness for measures, generalising the result above.

\subsection{Solovay tests in the quantum setting}
We discuss some quantum analogs of Solovay tests, a    test notion that is  equivalent to ML-tests in the classical setting \cite[Ch.\ 3]{Nies:book}. Quantum Solovay tests will be used in the statement of Theorem~\ref{thm:milleryu}.
\begin{definition}[Quantum Solovay randomness] \label{df:qSL}  \mbox{}
\bi \item A \emph{quantum Solovay  test}  is an effective  sequence $\seq {G_r}\sN r$ of   quantum $\SI 1$ sets such that $\sum_r \tau(G_r) < \infty$.
  \item We say that the test is \emph{strong} if the $G_r$ are given as projections; that is,  from $r$ we can compute $n_r$ and a   matrix of algebraic numbers  in   $\Malg {n_r}$ describing $G_r =  p_r$. 
%
\item For  $\delta \in (0,1)$, we     say that $\rho$ \emph{fails} the quantum  Solovay  test \emph{at order} $\delta$   if   $ \rho(G_r) > \delta$ for infinitely many  $r$; otherwise $\rho$ \emph{passes} the qML test \emph{at order} $\delta$. 
\item We say that $\rho$ is \emph{quantum Solovay-random}  if it passes each quantum Solovay  test  $\seq {G_r}\sN r$ at each positive order, that is,  $\lim_r \rho(G_r)=0$.  \ei
 \end{definition}
 Tejas Bhojraj  has shown that quantum \ML\  randomness implies quantum Solovay randomness; the converse implication is trivial.  He converts a quantum Solovay test into a quantum \ML\ test, so that failing the former at level $\delta$ implies failing the latter at level $O(\delta^2)$.  

\section{Initial segment complexity} \label{s:char}

Our definition of quantum \ML\ randomness   is by analogy with classical ML-randomness, but also based on the intuition that the properties of  quantum \ML\ random  states are hard to predict. So we expect that  the complexity of their initial segments is high. In order to formalize this, we start off from a theorem of G\'acs and also Miller and Yu~\cite{Miller.Yu:08} that asserts  that a sequence $Z$ is ML-random iff all its initial segments are hard to compress,  in the sense of plain descriptive string complexity. Our main result works towards an extension of this theorem to the quantum setting.

\subsection*{Classical setting.} Let $K(x)$ denote the prefix-free version of descriptive  complexity of a bit string $x$. (See \cite[Ch.\ 2]{Nies:book} for a brief  overview  of descriptive string complexity, also called Kolmogorov complexity.) The Levin-Schnorr theorem (see \cite[Thm.\ 5.2.3]{Downey.Hirschfeldt:book} or  \cite[Thm.\ 3.2.9]{Nies:book}) says that a bit sequence   $Z $ is ML-random if and only if  each of its initial segments is incompressible in the sense that $\ex b \in \NN  \,  \fa n \, K(Z\uhr n) > n-b$. 
The Miller-Yu Theorem \cite[Thm.\ 7.1]{Miller.Yu:08}  is a version of this in terms of  plain, rather than prefix-free,  descriptive string complexity, usually denoted by  $C(x)$. The constant $b$ is replaced by a sufficiently   fast growing computable function $f(n)$. We  will provide a quantum   analog of the Miller-Yu theorem, thereby avoiding the obstacles to  introducing prefix-free descriptive string complexity in the quantum setting. 

 A slight variant of the Miller-Yu Theorem was obtained by Bienvenu, Merkle and Shen~\cite{Bienvenu.Merkle.ea:08}. Their version states that, for an appropriate computable function $f$ such that $\sum \tp{-f(n)} < \infty$, $Z$ is ML-random iff there is $r$ such that for each $n$ we have $C(Z\uhr n \mid n) \ge n- f(n)-r$, where $C(x\mid n)$ is the plain Kolmogorov complexity of a string $x$ given its length $n$.   Requiring $\sum_n \tp{-f(n)}  < \infty$ of course means that $f$     grows  sufficiently fast; the borderline is between $\log_2 n$ and $2 \log_2 n$.

\subsection*{Quantum setting.} Quantum Kolmogorov complexity  is measured via quantum Turing machines~\cite{Bernstein.Vazirani:97,Vazirani:02}. In the version due to  Berthiaume, van Dam and Laplante \cite[Def.\ 7]{Berthiaume.ea:00},   the compression of a state of $\Malg n$ is via a  state of $\Malg k$, and only approximative in the sense that a state  in $\Malg n$ ``nearby'' the given state   can actually be compressed. More detail on this was provided in Markus M\"uller's thesis~\cite{Muller:07}, which in particular contains a detailed   discussion of how to define halting for a quantum Turing machine. 
   
In order to avoid obscuring the arguments below by discussions of  quantum Turing machines and universality, we will use a restricted machine model that is sufficient for a meaningful analog of the G\'acs-Miller-Yu Theorem. This machine model corresponds to uniformly generated circuit sequences.  After proving our result,  in Remark~\ref{rem: Schwachsinn}  we will discuss its relationship with quantum Kolmogorov complexity in the sense of~\cite{Berthiaume.ea:00}.

 \begin{convention} \label{conv} In the following all qubit sequences and all states will be elementary. Thus,   the relevant matrices only have   entries from the field  $\Calg$  of algebraic complex numbers. \end{convention} 

   \begin{definition}  \label{def:unitary machine} A \emph{unitary machine} $L$ is given by computable sequence of  unitary (elementary) matrices  $ L_n \in \Malg n$. For an input  $z $ which is a density matrix in $\Malg n$,     its     output is  $L(z ; n) := L_n z L_n^\dagger$.  Thus,   if $z $ is a pure state~$\ket {\psi}$, with the usual identifications the  output is $L_n \ket \psi$.    \end{definition}
  Recall that  the trace  norm  of an  $n\times n$    matrix $A$ over $\C$ is defined by $\norm A_{\mathsf tr}  = \sqrt{ \tr (A A^\dagger)}$. The  trace distance between two $n\times n$  matrices     is $D(A, B) = \frac 1 2 \norm{A - B}_{tr}$.
\begin{definition} \label{def:unitary QC} Let $L$ be a unitary machine.
  The $L$-quantum Kolmogorov complexity $QC_L^\eps(x \mid n)$ of a (possibly mixed)   state $x$ on $n$-qubits is the least natural number  $k$ such that there exists a  (mixed)     state $y \in \Malg k$ with \bc $D(x, L( y \otimes \ketbradouble {0^{n-k}}; n)) < \epsilon$. \ec That is,  the output of $L$  on $y \otimes \ketbradouble {0^{n-k}}$ approximates  $x$ to an accuracy of $ \eps$  in  the trace distance.    \end{definition}

We fix a computable listing $\seq{\sss_i}\sN i$ of the elementary pure qubit strings so that $\ell(\sss_i) \le i$ for each~$i$.

 We now   prove a weak quantum analog  of the G\'acs-Miller-Yu theorem. %
%
%
 

%
%

\begin{theorem}\label{thm:milleryu} {\rm 
  Let $\rho$ be a state on $\Malg \infty$. 
  \begin{enumerate} 
    \item   Let $L$ be a unitary machine. Let $1 > \eps >0$ and suppose $\rho$ passes each qML-test at order $1-\eps$. Then for each computable function $f$ satisfying $\sum_n \tp{-f(n)}  < \infty$,  for almost every $n$
       \begin{align*}
         QC_L^{\eps}(\rho\uhr n \mid n ) \geq n- f(n). 
       \end{align*}
     
    \item For each strong  quantum Solovay test $\seq{p_r}\sN r$, there exists a total computable function $f \colon\,  \NN \to \NN$ with  $\sum_n \tp{-f(n)}  \le  4$  and a unitary machine $L$ such that the  following holds. If  $\rho$ fails $\seq {p_r}$   at order $1-\eps$ where   $1 > \eps >0$, then   there are infinitely many   $n$ such that          \begin{align*} 
        QC_L^{\sqrt \eps} (\rho\uhr n \mid n ) < n- f(n).
      \end{align*}

  \end{enumerate}
  }
\end{theorem}
Note that this is not a full analog,  because the second part can only be obtained under the hypothesis that $\rho$  fails a strong Solovay test (Def.\ \ref{df:qSL}).
  Bhojraj has announced an alternative version of Part 1 where  the hypothesis is that   $\rho$  pass all strong Solovay tests at order $1 - \eps$. 


%
%
%
%

\begin{proof} 
\emph{Part~1:}
We may assume that $\sum_n \tp{-f(n)}\le 1/4$, because we can replace $f$ by $\widetilde f = f+C$ where $C$ is sufficiently large so that this condition is met. 

Recall  our fixed listing $\seq{\sss_i }\sN i$ of the pure elementary quantum states of any length.
 For a given parameter $r \in \NN$,  and   $t,n\in \NN$, let $S_{r,t}(n)$ be the set of  pure qubit strings $x= \sss_i$, $i \le t$, of length   $n$ so that 
for some pure qubit string $y = \sss_k$, $k \le t$, we have  \bc   $|y| \le n- f(n)-r$ and $L(y\otimes \ket{0^{n-|y|}}; n)=x$. \ec

Note that $S_{r,t}(n)$ is computable in $r,t,n$ using Convention~\ref{conv}. Hence from $r,t,n$  we can compute   an orthogonal projection $p_{r,t}(n)$  in $\Malg n$ onto the subspace generated by $S_{r,t}(n)$. Let $p_{r,t} = \sup_{n \le t }p_{r,t}(n)$. Then $p_{r,t} \in \Malg t$ and $p_{r,t}$ is computable in $r,t$. Clearly $p_{r,t} \le p_{r, t+1}$ for each $t$.

   By definition of unitary machines  the dimension of the range of $p_{r,t}(n)$ is bounded by $\tp{n- f(n) -r+2}$. Hence $\tau(p_{r,t}(n)) \le \tp{-f(n)-r+2}$, and then $\tau(p_{r,t}) \le \sum_n \tp{-f(n)-r+2} \le \tp{-r}$ by our hypothesis on $f$.

Let $G_r$ be the quantum $\SI 1$ set given by the sequence $\seq {p_{r,t} }\sN t$. Then $\seq{G_{r}}\sN r$ is a quantum ML-test.

We    show that there is $r$ such that for each $n$ we have $QC_L^{\eps}(\rho\uhr n \mid n ) \geq n- f(n)-r 
$. Since we can carry out the same argument with $\lfloor \frac 1 2 f \rfloor$ instead of $f$,  and $f(n) \to \infty$, this will be sufficient. 

We proceed by contraposition. Suppose that for  arbitrary $r\in \NN$  there is $n$ such that   $QC_L^\eps(\rho \uhr n \mid n) < n- f(n) -r$. This means that there is a state  $y$ (possibly   mixed)  of length $k < n- f(n) -r$  such that  $D(x, \rho\uhr n)< \eps$ where $x= L(y \otimes \ketbradouble{0^{n-k}},n)$.  Let  $y = \sum \alpha_i  \ketbra {y_i}{y_i}$ be  the corresponding  convex combination of pure states with $\alpha_i$ algebraic and $y_i$ of length  $k$. We have  $x= \sum_i \alpha_i \ketbra {x_i}{x_i}$ where $x_i= L(y_i \otimes \ket{0^{n-k}},n)$.   Then there is $t$  such that $x_i\in S_{r,t}(n)$ for each $i$, and hence $\tr {[x p_{r,t}}(n)]=1$. This implies that  $\rho(p_{r,t}(n)) > 1- \eps$ and hence $\rho(G_r) > 1-\eps$. Since $r$ was arbitrary this shows that  $\rho$ fails the test at order $1-\eps$.

 \medskip

\n
  \emph{Part 2:}   Let  $\seq {p_r}\sN r$ be a strong quantum Solovay test. We may assume that $\sum_r \tau(p_r) \le 1/2$, and that $p_r \in  \Malg {n_r}$ where $n_r$ is computed from $r$  and  $n_r<n_{r+1}$ for each $r$. 
  The idea is as follows:   suppose the range of $p_r$ has   dimension $k< n_r$. Then  $z_n'$, the   projection of $\rho\uhr n$ to $p_r$ as defined below,  can be directly described by a density matrix in   $\Malg k$ if we define  our unitary machine $L$ to  compute an isometry between $\+ H_k$ and the range of $p_r$.  If  $\rho  (p_r) > 1-\eps$ we show that  the trace distance from $z_n'$ to $\rho\uhr n$ is at most~$\sqrt \epsilon$. Therefore  $QC_L^{\sqrt \eps} (\rho\uhr n \mid n ) \le k$.  For a function $f$ as required, we can ensure $k  < n - f(n_r)$ for each $r$.

 For the details, let  $f \colon \NN \to \NN $ be a computable  function such that 
       \begin{align*}
    2^{-f(n_r)}  \ge   \tau (p_r) >  2^{-f(n_r)-1} 
  \end{align*}
   and  $f(m) = m$ if  $m$ is not of the form $n_r$. Note that $f$ is computable and satisfies $\sum_n \tp{-f(n)} \le 4$. 
  Let $g(n) = n- f(n)$.
  
  To  describe the unitary machine $L$ we need to provide  a computable sequence of unitary matrices $\seq{L_n}\sN n$.  For $n = n_r$, let $L_n$ be a unitary matrix  in $\Malg n$ such that its restriction $L_n \upharpoonright {\+ H_{g(n)}}$ is an isometry $\+ H_{g(n)} \cong \range  (p_r) $ (the range of $p_r$). By  hypothesis on the sequence $\seq {p_r}$ this sequence of unitary matrices is computable. 

 For a projection operator  $p$ in $\Malg n$ and    a density matrix $s$  in $\Malg n$, we define the projection of $s$ by $p$   to be
\begin{align*} \Proj s p   = \frac{1}{\tr[sp]} psp \end{align*}
Note that this is again a density matrix, and each of its eigenvectors is  in the range of $p$.  

In the following fix an $r$ such that   $\rho  (p_r) > 1-\eps$.  
Write  $n= n_r$ and  $z_n = \rho \uhr  n\in \Malg n$. So $\tr( z_n  p_r) > 1-\eps$.    Let 
%
\bc $z'_n= \Proj  {z_n} {p_{r}} $  \ec
%
 
 %

\begin{claim}  \label{cl:ott} $QC_L^\delta (z'_n \mid n) \le g(n)$ for each $\delta >0$.
\end{claim}
\begin{proof} Each eigenvector of $z_n'$ is in the range of $p_r$. So there is a density matrix $y \in \Malg {g(n)}$ such that  $L_n  (y \otimes \ketbradouble{0^{\otimes f(n)}}) L_n^\dagger= z_n'$.
\end{proof} 

   We now argue that      $D(z'_n,z_n) < \sqrt \eps$ (recall that $D$ denotes  the trace distance).  We rely on the following.
  
  \begin{proposition} \label{fa:fid} Let $p$ be a projection in $\Malg n$, and let $\theta$ be a density matrix  in $\Malg n$.  Write $\alpha = \tr { [ \theta p]}$. Let $\theta ' = \Proj \theta p$. Then  $D (\theta', \theta) \le \sqrt {1- \alpha}$.  \end{proposition}
  \begin{proof}  Let $\ket{\psi_\theta}$ be a  purification of $\theta$.  Then $\alpha^{-\half} p  \ket{\psi_\theta}$ is a purification of~$\theta'$.    Uhlmann's theorem (e.g.\ \cite[Thm.\ 9.4]{Nielsen.Chuang:02}) implies  \bc $F(\theta',\theta) \geq \alpha^{-\half} \Scp{\psi_\theta}{p \mid \psi_\theta} = \alpha^\half$,  \ec where $F(\sss, \tau)  =\tr {\sqrt {\sqrt \sss \tau \sqrt \sss}}$ denotes fidelity.  
  Now it suffices to recall  from e.g.\ \cite[Eqn.\ 9.110]{Nielsen.Chuang:02}  that $D(\theta', \theta) \le \sqrt {1- F(\theta', \theta)^2}$. 
  \end{proof} 
 We apply  Prop.~\ref{fa:fid} to   $p = p_{r}$  and $\theta = z_n$ and $\theta' = z'_n$, where as above $n= n_r$.  By hypothesis $\alpha = \tr {[z_n p_{r} ]} =  \rho  (p_{r}) > 1-\eps$ and hence $\sqrt {1- \alpha} < \sqrt \eps$. Claim~\ref{cl:ott} now  shows $QC_L^\epsilon (z_n) \le g(n)$.
 %
 Since there are infinitely many $r$ such that    $\rho  (p_r) > 1-\eps$, we obtain  Part 2 of Thm.\ \ref{thm:milleryu}.
\end{proof}

\begin{remark}\label{rem: Schwachsinn} 
  {\rm In an important extended abstract, Yao~\cite{Yao:93}  proved that the quantum Turing machines  (QTM) of  Bernstein and Vazirani~\cite{Bernstein.Vazirani:97} can be simulated by  quantum circuits with only a polynomial overhead in time. For  recent work  on such a  simulation see  \cite{Molina.Watrous:19}. 

Yao  also announced    the converse direction: a QTM can simulate the input/output behavour of a  computable sequence of  quantum circuits. The argument is briefly discussed after the statement of Theorem 3  in~\cite{Yao:93} (also see \cite{Westergaard:05}).  So with suitable input/output conventions, Definition~\ref{def:unitary QC} can be seen as  a special case of the definition of $QC^\eps_M$ for a quantum Turing machine $M$ as in~\cite[Def.\ 7]{Berthiaume.ea:00}.

We ignore at present whether Part 1 can be strengthened to general quantum Turing machines. The input/output behaviour of such a machine is merely given by a quantum operation, for instance because at the end of a computation the state has to be discarded. }
   
   \end{remark}

On the other hand, since a QTM $M$ can simulate the effect of the sequence  of quantum circuits $\seq {L_n}$, and a universal QTM in the sense of~\cite{Bernstein.Vazirani:97} can simulate $M$ with a small loss in accuracy, we obtain the following.

\begin{cor} In the setting of Part 2 of the theorem, we have $QC^{2\sqrt \eps} (\rho\uhr n \mid n ) < n- f(n)$ for infinitely many $n$. \end{cor}
 It would   be interesting to find a version of   Theorem~\ref{thm:milleryu} in terms of G\'acs' version of quantum Kolmogorov complexity, which is based on semi-density matrices rather than machines \cite{Gacs:01}.

\section{Outlook}
  As mentioned,   randomness via algorithmic tests  has   been related to effective dynamical systems in   papers such as~\cite{Nandakumar:08,Hoyrup:12}. 
For a promising connection with quantum information processing, recall that  a spin chain can be seen as  a quantum dynamical system with the shift operation \cite{Bjelakovic.etal:04}, in analogy with the classical case with  the shift oepration on $\cantor$ that deletes the first bit of a sequence. An interesting potential application of our randomness notion is to obtain an effective  quantum version of the   Shannon-McMillan-Breiman  (SMB) theorem from the 1950s (see e.g.\ \cite{Shields:96}). That result is important in the area of data compression because  it determines the asymptotic compression rate of sequences of symbols emitted by an ergodic  source.      

Let $A$ be a finite alphabet, and let $P$ be an ergodic probability measure on $A^\NN$.  Let $h(P)$ denote the entropy of $P$. The classical theorem states that  for almost every $Z\in A^\NN$,  we  have
\begin{equation} \label{eqn: hP} h(P) = - \lim_n \frac 1 n \log P[ Z\uhr n]; \end{equation}   informally, $h(P)$ can be obtained by looking at the asymptotic  ``empirical entropy'' along a  random sequence~$Z$. 
  
   The Shannon-McMillan (SM)  theorem is a slightly  weaker, earlier  version of the full SMB theorem based on the notion of  convergence in probability.   Breiman then  phrased the   Shannon-McMillan theorem as a property of almost every sequence in the sense of the given ergodic measure.  Bjelakovic et al.\   \cite{Bjelakovic.etal:04} provided   a quantum version of the SM-theorem. Their setting is the one of  bi-infinite spin chains, or more generally, $d$-dimensional lattices; it   can be easily adapted to the present setting.   %
  
  Algorithmic  versions
of the SMB theorem~\cite{Hochman:09, Hoyrup:12}       show that \ML\ randomness  of $Z$ relative to the computable ergodic  measure is sufficient for (\ref{eqn: hP}) to hold.  The question then is whether in the quantum setting, quantum ML-randomness relative to a computable shift-invariant ergodic state  is sufficient.

 The von Neumann entropy of a density matrix $S$ is defined by  $H(S) = - \tr (S \log S)$. For a state $\psi$ on $\Malg \infty$ we let 
$h(\psi) =  \lim_n \frac  1 n H(\psi \uhr {\Malg n}$. For background and  notions not defined here see \cite{LogicBlog:17}, Section 6.
\begin{conjecture}[with M.\ Tomamichel]  
 Let $\psi$ be an ergodic computable state on $\Malg \infty$. Let $\rho$ be a state that is quantum ML-random   with respect to $\psi$. Then 
$ h(\psi) = -  \lim_n \frac 1 n\tr (\rho \uhr  {\Malg { {n}}}  \log  ( \psi \uhr  {\Malg{ {n}}}))$.  \end{conjecture}
  Note that this reduces to the classical theorem in case $\psi$ is a probability measure and $\rho$ a bit sequence, because each matrix $\rho \uhr  {M_{ {n}}}  \log ( \psi \uhr  {M_{ {n}}})$  is diagonal with at most one nonzero entry. Tomamichel and the first author have verified the conjecture in case $\mu$ is an i.i.d.\ state; see~\cite[Section 6]{LogicBlog:17}. The first author and Stephan \cite[Thm.\ 23]{Nies.Stephan:19} have proved the conjecture in case that $\psi$ and $\rho$ are measures on Cantor space  and the empirical entropy $-\frac 1 n \log \psi[x]$, for $x$ an $n$-bit string, is bounded above (this means that    $\psi$ is close to the  uniform measure).

%


%
 
\medskip
\n \emph{Acknowledgements.} We thank Tejas Bhojraj and Marco Tomamichel for corrections and helpful discussions,  and the anonymous referee for many helpful suggestions, in particular on Theorem~\ref{thm:milleryu}. We also thank  Willem Fouch\'e and Peter G\'acs   for helpful comments. 

VBS is grateful for the hospitality of the University of Auckland. This work was started while he was still with the Department of Physics at Ghent University, and supported by the EU through the ERC grant Qute. He also acknowledges support by the NCCR QSIT.  AN is grateful for the hospitality of Ghent University and  of the IMS at NUS, Singapore. He  acknowledges support by the Marsden fund of New Zealand. 

%
%
%

 \def\cprime{$'$} \def\cprime{$'$}

\end{document}

%% file: volkherincludes.tex
\def\C{{\mathbb C}}     
\def\norm #1{\Vert #1\Vert}

\def\ket #1{\vert #1\rangle}

\def\ketbra #1#2{\vert #1\rangle \langle #2\vert}

\newcommand\Scp[2]{\ensuremath{\, \langle #1 \,\vert #2 \,\rangle}}

\def\S{{\mathcal S}}

\newcommand*{\eps}{\varepsilon}

\newcommand*{\half}{\frac{1}{2}}


%% file: NiesScholz_2019.bbl
\begin{thebibliography}{10}

\bibitem{Bausch.etal:18}
J.~Bausch, T.~Cubitt, A.~Lucia, D.~Perez-Garcia, and M.~Wolf.
\newblock Size-driven quantum phase transitions.
\newblock {\em Proceedings of the National Academy of Sciences}, 115(1):19--23,
  2018.

\bibitem{Benatti.etal:14}
F.~Benatti, S.~Oskouei, and A.~Deh~Abad.
\newblock Gacs quantum algorithmic entropy in infinite dimensional {H}ilbert
  spaces.
\newblock {\em Journal of Mathematical Physics}, 55(8):082205, 2014.

\bibitem{Bernstein.Vazirani:97}
E.~Bernstein and U.~Vazirani.
\newblock Quantum complexity theory.
\newblock {\em SIAM Journal on computing}, 26(5):1411--1473, 1997.

\bibitem{Berthiaume.ea:00}
A.~Berthiaume, W.~Van~Dam, and S.~Laplante.
\newblock Quantum {K}olmogorov complexity.
\newblock In {\em Computational Complexity, 2000. Proceedings. 15th Annual IEEE
  Conference on}, pages 240--249. IEEE, 2000.

\bibitem{Bienvenu.Merkle.ea:08}
L.~Bienvenu, W.~Merkle, and A.~Shen.
\newblock A simple proof of the {M}iller-{Y}u {T}heorem.
\newblock {\em Fundamenta Informaticae}, 83(1-2):21--24, 2008.

\bibitem{Bjelakovic.etal:04}
I.~Bjelakovi{\'c}, T.~Kr{\"u}ger, R.~Siegmund-Schultze, and A.~Szko{\l}a.
\newblock The {S}hannon-{M}c{M}illan theorem for ergodic quantum lattice
  systems.
\newblock {\em Inventiones mathematicae}, 155(1):203--222, 2004.

\bibitem{Bratteli.Robinson:12}
O.~Bratteli and D.~Robinson.
\newblock {\em Operator Algebras and Quantum Statistical Mechanics: Volume 1:
  C*-and W*-Algebras. Symmetry Groups. Decomposition of States}.
\newblock Springer Science \& Business Media, 2012.
\newblock First edition 1979.

\bibitem{Cubitt.etal:15}
T.~Cubitt, D.~Perez-Garcia, and M.~Wolf.
\newblock Undecidability of the spectral gap.
\newblock {\em Nature}, 528(7581):207--211, 2015.
\newblock Full version arXiv:1502.04573, 146 pages.

\bibitem{Downey.Hirschfeldt:book}
R.~Downey and D.~Hirschfeldt.
\newblock {\em Algorithmic randomness and complexity}.
\newblock Springer-Verlag, Berlin, 2010.
\newblock 855 pages.

\bibitem{LogicBlog:17}
A.~Nies (editor).
\newblock Logic {B}log 2017.
\newblock Available at \url{http://arxiv.org/abs/1804.05331}, 2017.

\bibitem{Furber:19}
R.~Furber.
\newblock Projections in {C}{A}{R} (canonical anticommutation relation)
  algebra.
\newblock MathOverflow.
\newblock URL:https://mathoverflow.net/q/334640 (version: 2019-06-24).

\bibitem{Gacs:80}
P.~G{\'a}cs.
\newblock Exact expressions for some randomness tests.
\newblock {\em Z. Math. Logik Grundlag. Math.}, 26(5):385--394, 1980.

\bibitem{Gacs:01}
P.~G{\'a}cs.
\newblock Quantum algorithmic entropy.
\newblock In {\em Computational Complexity, 16th Annual IEEE Conference on,
  2001.}, pages 274--283. IEEE, 2001.

\bibitem{Hochman:09}
M.~Hochman.
\newblock Upcrossing inequalities for stationary sequences and applications.
\newblock {\em Annals of Probability}, 37(6):2135--2149, 2009.

\bibitem{Hoyrup:12}
M.~Hoyrup.
\newblock The dimension of ergodic random sequences.
\newblock In Christoph D\"urr and Thomas Wilke, editors, {\em STACS}, pages
  567--576, 2012.

\bibitem{Levin:74}
L.~Levin.
\newblock Laws of information conservation (nongrowth) and aspects of the
  foundation of probability theory.
\newblock {\em Problemy Peredachi Informatsii}, 10(3):30--35, 1974.

\bibitem{Martin-Lof:66}
P.~Martin-L{\"o}f.
\newblock The definition of random sequences.
\newblock {\em Inform. and Control}, 9:602--619, 1966.

\bibitem{Miller.Yu:08}
J.~Miller and L.~Yu.
\newblock On initial segment complexity and degrees of randomness.
\newblock {\em Trans. Amer. Math. Soc.}, 360:3193--3210, 2008.

\bibitem{Molina.Watrous:19}
A.~Molina and J.~Watrous.
\newblock Revisiting the simulation of quantum turing machines by quantum
  circuits.
\newblock {\em Proceedings of the Royal Society A}, 475(2226):20180767, 2019.

\bibitem{Monin.Nies:15}
B.~Monin and A.~Nies.
\newblock A unifying approach to the {G}amma question.
\newblock In {\em Proceedings of Logic in Computer Science (LICS)}. IEEE press,
  2015.

\bibitem{Monin.Nies:17}
B.~Monin and A.~Nies.
\newblock Muchnik degrees and cardinal characteristics.
\newblock {\em arXiv preprint arXiv:1712.00864}, 2017.

\bibitem{Muller:07}
M.~M{\"u}ller.
\newblock Quantum {K}olmogorov complexity and the quantum {T}uring machine.
\newblock {\em PhD Thesis, University of Berlin, arXiv:0712.4377}, 2007.
\newblock PhD Thesis, University of Berlin.

\bibitem{Nandakumar:08}
S.~Nandakumar.
\newblock An effective ergodic theorem and some applications.
\newblock In {\em Proceedings of the fortieth annual ACM symposium on Theory of
  computing}, pages 39--44. ACM, 2008.

\bibitem{Nielsen.Chuang:02}
M.~Nielsen and I.~Chuang.
\newblock {\em Quantum computation and quantum information}.
\newblock AAPT, 2002.

\bibitem{Nies:book}
A.~Nies.
\newblock {\em Computability and {R}andomness}, volume~51 of {\em Oxford Logic
  Guides}.
\newblock Oxford University Press, Oxford, 2009.
\newblock 444 pages. Paperback version 2011.

\bibitem{Nies.Stephan:19}
A.~Nies and F.~Stephan.
\newblock A weak randomness notion for measures.
\newblock Available at \url{https://arxiv.org/abs/1902.07871}, 2019.

\bibitem{Pathak.Rojas.ea:12}
N.~Pathak, C.~Rojas, and S.~G. Simpson.
\newblock Schnorr randomness and the {L}ebesgue differentiation theorem.
\newblock {\em Proc. Amer. Math. Soc.}, 142(1):335--349, 2014.

\bibitem{Schnorr:71}
C.P. Schnorr.
\newblock A unified approach to the definition of random sequences.
\newblock {\em Mathematical Systems Theory}, 5(3):246--258, 1971.

\bibitem{Shields:96}
P.~Shields.
\newblock {\em The Ergodic Theory of Discrete Sample Paths}.
\newblock Graduate Studies in Mathematics 13. American Mathematical Society,
  1996.

\bibitem{Vazirani:02}
U.~Vazirani.
\newblock A survey of quantum complexity theory.
\newblock In {\em Proceedings of Symposia in Applied Mathematics}, volume~58,
  pages 193--220, 2002.

\bibitem{Westergaard:05}
C.~Westergaard.
\newblock Computational equivalence between quantum turing machines and quantum
  circuit families.
\newblock {\em University of Copenhagen, Denmark, available on Semantic
  Scholar}, 2005.

\bibitem{Yao:93}
A.~Yao.
\newblock Quantum circuit complexity.
\newblock In {\em Proceedings of 1993 IEEE 34th Annual Foundations of Computer
  Science}, pages 352--361. IEEE, 1993.

\end{thebibliography}
